\documentclass[11pt]{article}
\usepackage{fullpage}
\usepackage[margin=2.5cm]{geometry}

\usepackage[utf8]{inputenc}

\usepackage{amsmath, amssymb, amsthm}
\usepackage{enumitem}

\usepackage{graphicx}
\usepackage{color}
\usepackage[dvipsnames]{xcolor}

\usepackage{natbib}
\usepackage{nicefrac}

\usepackage{amsmath}
\usepackage{mdframed}

\usepackage[T1]{fontenc}

\newtheorem{theorem}{Theorem}[section]
\newtheorem{claim}[theorem]{Claim}

\newtheorem{lemma}[theorem]{Lemma}

\newcommand{\E}{\mathbf{E}}

\newcommand{\calF}{\mathcal{F}}

\DeclareMathOperator*{\argmin}{argmin}
\title{Single-Pass Pivot Algorithm for Correlation Clustering. \\
Keep it simple!}
\author{Sayak Chakrabarty \and Konstantin Makarychev}
\date{Northwestern University}
\begin{document}
\maketitle
\begin{abstract}
   We show that a simple single-pass semi-streaming variant of the Pivot algorithm for Correlation Clustering gives a $(3+\varepsilon)$-approximation using $O(n/\varepsilon)$ words of memory. This is a slight improvement over the recent results of Cambus, Kuhn, Lindy, Pai, and Uitto, who gave a $(3+\varepsilon)$-approximation using $O(n \log n)$ words of memory, and Behnezhad, Charikar, Ma, and Tan, who gave a 5-approximation using $O(n)$ words of memory. One of the main contributions of this paper is that both the algorithm and its analysis are very simple, and also the algorithm is easy to implement.
\end{abstract}
\section{Introduction}
In this paper, we provide a simple single-pass streaming algorithm for Correlation Clustering. The Correlation Clustering model was  introduced by \cite*{bansal2004correlation}. In this model, we have a set of objects represented by vertices of a graph. A noisy binary classifier labels all edges of the graph with ``+'' and ``-'' sings. We call edges labeled with a ``+'' sign -- positive edges and edges labeled with a ``-'' sign -- negative edges. A positive edge indicates that the endpoints of the edge are similar, and a negative edge indicates that the endpoints are dissimilar. Given a labeled graph, our goal is to find a  clustering that disagrees with the minimum number of edge labels. A positive edge disagrees with the clustering if its endpoints are in different clusters, and a negative edge disagrees with the clustering if its endpoints are in the same cluster. We assume that the classifier can make many errors and, consequently, there may be no clustering that agrees with  all labels of the graph.

Correlation Clustering has been widely studied in ML and theoretical computer science communities. It was 
used for image-segmentation (\cite{kim2014image}), link prediction (\cite{yaroslavtsev2018massively}), document clustering (\cite{bansal2004correlation}), community detection (\cite{veldt2018correlation,shiscalable}), and other problems. 

In Section \ref{sec:related-work}, 
we discuss various objectives and assumptions used for Correlation Clustering. We now focus on the most common setting, where the graph $G$ is a complete graph (i.e., we have a ``+'' or ``-'' label for each pair of vertices $u$ and $v$), the classifier can make  arbitrary mistakes, and the objective is to minimize the number of disagreements.

In their original paper, \cite*{bansal2004correlation} gave a constant factor approximation for Correlation Clustering on complete graphs. \cite*{charikar2005clustering} improved the approximation factor to 4. Then, \cite*{ailon2008aggregating} proposed combinatorial and LP-based \textsc{Pivot} algorithms, with approximation factors $3$ and $2.5$, respectively. The LP-based \textsc{Pivot} algorithm was further improved by 
\cite*{chawla2015near}, who provided a $2.06$ approximation. Finally, \cite*{cohen2022correlation} used a Sherali-Adams relaxation for the problem to get a $1.994 + \varepsilon$ approximation factor. On the hardness side, 
\cite*{charikar2005clustering} proved that the problem is APX-Hard and thus unlikely to admit a PTAS. 

Most abovementioned algorithms use linear programming (LP) for finding a clustering. The number of constraints in the standard LP is $\Theta(n^3)$ (due to ``triangle inequality constraints''). Solving such LPs is prohibitively expensive for large  data sets\footnote{\cite{veldt2022correlation} recently proposed a practical linear programming algorithm for Coreallation Clustering.}. That is why, in practice,  the algorithm of choice is the combinatorial  \textsc{Pivot} algorithm by
\cite{ailon2008aggregating} along with its variants. The approximation factor of this algorithm is 3.

To deal with massive data sets, researchers started exploring parallel (MPC) and semi-streaming algorithms. Recent works on MPC algorithms include 
\cite*{blelloch2012greedy,chierichetti2014correlation,pan2015parallel,ahn2015correlation,
ghaffari2018improved,fischer2019tight,cambus2021massively,cohen2021correlation,assadi2022sublinear,cambus2022parallel,behnezhad2022almost}.

Streaming algorithms have been extensively studied by 
\cite*{chierichetti2014correlation,ahn2015correlation,
ghaffari2018improved,cohen2021correlation,assadi2022sublinear,cambus2022parallel,behnezhad2022almost,behnezhad2023single}. This year, \cite*{behnezhad2023single} designed a single-pass polynomial-time semi-streaming $5$-approximation algorithm that uses $O(n)$ words of memory. Then, \cite*{cambus2022parallel} provided a single-pass polynomial-time semi-streaming ($3+\varepsilon$)-approximation algorithm that uses $O(n\log n)$ words of memory. Both algorithms are based on the combinatorial \textsc{Pivot} algorithm. 

\paragraph{Our Results.} In this paper, we examine a simple semi-streaming variant of the \cite*{ailon2008aggregating} combinatorial \textsc{Pivot} algorithm that uses $O(n/\varepsilon)$ words of memory and gives a $3+\varepsilon$ approximation. Our algorithm needs less memory than the algorithm by \cite*{cambus2022parallel} and gives the same $(3+\varepsilon)$ approximation guarantee. The algorithm is very simple and easy to implement. Its analysis is simple and concise.

\paragraph{Semi-streaming Model.} We consider the following single-pass semi-streaming model formally defined by~\cite*{feigenbaum2005graph} (see also~\cite*{muthukrishnan2005data}). The algorithm gets edges of the graph along with their labels from an input stream. The order of edges in the graph is arbitrary and can be adversarial. The algorithm can read the stream only once. We assume that the algorithm has $O(n)$ words of memory, where $k$ is a constant. Thus, it cannot store all edges in the main memory. As in many previous papers on streaming algorithms for Correlation Clustering, we shall assume that the stream contains only positive edges (if negative edges are present in the stream, our algorithm will simply ignore them). We prove the following theorem.

\begin{theorem}
Algorithm 1 (see Figure~\ref{fig:main-streaming-alg}) is a randomized polynomial-time semi-streaming algorithm with an approximation factor $3+O(1/k)$. The algorithm uses $O(kn)$ words of memory, where each word can store numbers between $1$ and $n$. The algorithm spends $O(\log k)$ units of time for every edge it reads in the streaming phase. Then, it requires additional $O(kn)$ units of time to find a clustering (in \emph{the pivot selection and clustering} phase).
\end{theorem}

\subsection{Related Work}\label{sec:related-work}

In this paper, we consider the most commonly studied objective, \textsc{MinDisagree}, of minimizing the number of disagreements. \cite*{swamy2004correlation} and \cite*{charikar2005clustering} 
investigated a complementary objective, \textsc{MaxAgree}, of maximizing the number of agreements. They gave $0.766$-approximation algorithms for \textsc{MaxAgree}. \cite*{puleo2016correlation,chierichetti2017fair,charikar2017local,ahmadi2019min,kalhan2019correlation,ahmadian2020fair,friggstad2021fair,jafarov2021local, schwartz2022fair, ahmadian2023improved,davies2023fast} considered different variants of fair and local objectives.

\cite*{charikar2005clustering} and \cite*{demaine2006correlation} provided $O(\log n)$ approximation for Correlation Clustering on incomplete graphs and edge-weighed graphs. For the case when all edge weights are in the range $[\alpha,1]$, \cite*{jafarov2020correlation} gave a $(3+2\ln\nicefrac{1}{\alpha})$-approximation algorithm.

The problem has also been studied under the assumption that the classifier makes random or semi-random errors by \cite*{bansal2004correlation,elsner2009bounding,mathieu2010correlation,makarychev2015correlation}.
Online variants of Correlation Clustering were considered by \cite*{mathieu2010online,lattanzi2021robust,cohen2022online}.

\section{Single-Pass Streaming Algorithm}
Our algorithm is an extension of the algorithm by
\cite*{behnezhad2023single}. Loosely speaking, the algorithm works as follows: It picks a random ranking of vertices, scans the input stream and keeps only the $k$ top-ranked neighbours for every vertex $u$. Then, it runs the~\cite*{ailon2008aggregating} \textsc{Pivot} algorithm on the graph that contains only edges from every vertex to its $k$ top-ranked neighbours. We provide the details below.

We shall assume that every vertex is connected with itself by a positive edge. In other words, we include each vertex in the set of its own (positive) neighbours. We denote the set of positive neighbours of $u$ by $N(u)$.
The algorithm first picks a random ordering of vertices $\pi:V\to \{1,\dots,n\}$. We say that vertex $u$ is ranked higher than vertex $v$ according to $\pi$ if $\pi(u)<\pi(v)$. Therefore, $\pi^{-1}(1)$ is the highest ranked vertex, and $\pi^{-1}(n)$ is the lowest ranked vertex.  For every vertex $u$, we will keep track of its $k$ highest-ranked neighbours. We denote this set by $A(u)$. We initialize each $A(u) = \{u\}$. Then, for every positive edge $(u,v)$ in the input stream, we add $u$ to the set of $v$'s neighbours $A(v)$ and $v$ to the set of $u$'s neighbours $A(u)$. If  $A(u)$ or $A(v)$ get larger than $k$, we remove the lowest ranked vertices from these sets to make sure that the sizes of $A(u)$ and $A(v)$ are at most $k$. After we finish reading all edges from the input stream, we run the following variant of the \textsc{Pivot} algorithm. We mark all vertices as non-pivots.
Then, we consider vertices according to their ranking (i.e.,
$\pi^{-1}(1),\pi^{-1}(2),
\dots,\pi^{-1}(n)$). For every vertex $u$,  we find the highest ranked neighbour $v$ in $A(u)$, such that $v = u$ or $v$ is a \emph{pivot}. If such $v$ exists, we place $u$ in the cluster of $v$. Moreover, if $v=u$, we declare $u$ a \emph{pivot}. If no such $v$ in $A(u)$ exists, we put $u$ in its own cluster and declare it a \emph{singleton}. We provide pseudo-code for this algorithm in Figure~\ref{fig:main-streaming-alg}.

\begin{figure}
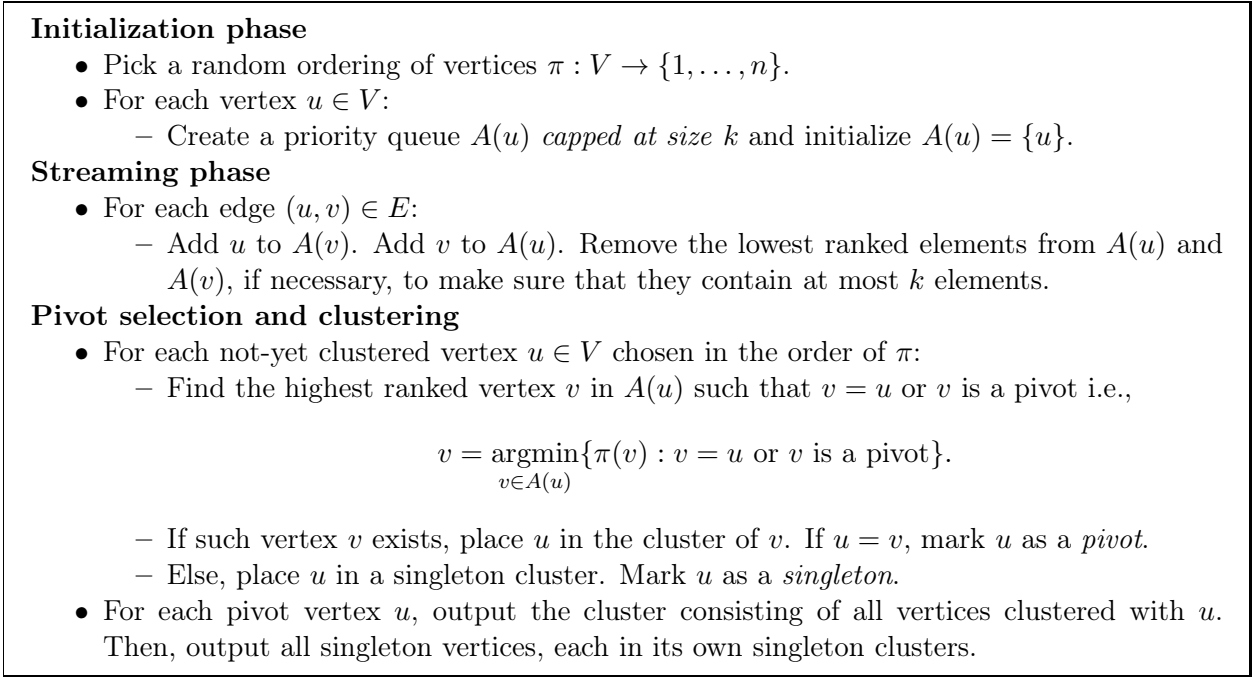

\begin{mdframed}[linewidth=0.5pt]
\textbf{Initialization phase}
\begin{itemize}[noitemsep,nolistsep]
\item Pick a random ordering of vertices $\pi: V\to\{1,\dots,n\}$.
\item For each vertex $u \in V$:
\begin{itemize}[noitemsep,nolistsep]
    \item Create a priority queue $A(u)$ \emph{capped at size} $k$ and initialize $A(u) = \{u\}$.
\end{itemize}
\end{itemize}      
\textbf{Streaming phase}
\begin{itemize}[noitemsep,nolistsep]
\item For each edge $(u, v) \in E$:
\begin{itemize}[noitemsep,nolistsep]
    \item Add $u$ to $A(v)$. Add $v$ to $A(u)$. Remove the lowest ranked elements from $A(u)$ and $A(v)$, if necessary, to make sure that they contain at most $k$ elements.
\end{itemize}
\end{itemize}
\textbf{Pivot selection and clustering}
\begin{itemize}[noitemsep,nolistsep]
\item For each not-yet clustered vertex $u \in V$ chosen in the order of $\pi$:  
\begin{itemize}[noitemsep,nolistsep]
\item Find the highest ranked vertex $v$ in $A(u)$  such that $v = u$ or $v$ is a pivot i.e.,
$$
v =
\argmin_{v\in A(u)}\{\pi(v): v = u \text{ or } v \text{ is a pivot}\}.
$$
\item If such vertex $v$ exists, place $u$ in the cluster of $v$. If $u = v$,
mark $u$ as a \emph{pivot}.
\item Else, place $u$ in a singleton cluster. Mark $u$ as a \emph{singleton}.
\end{itemize}
\item For each pivot vertex $u$, output the cluster consisting of all vertices clustered with $u$. Then, output all singleton vertices, each in its own singleton clusters.
\end{itemize}
\end{mdframed}
\caption{Semi-Streaming Algorithm for Correlation Clustering}
\label{fig:main-streaming-alg}
\end{figure}

We use fixed-size priority queues to store sets $A(u)$. Whenever we need to add a vertex $v$ to $A(u)$, we insert $v$ into $A(u)$ and remove the lowest-ranked element from $A(u)$ if the size of $A(u)$ is greater than~$k$. We remark that even though each $u$ is always added to set $A(u)$ in the beginning of the algorithm, it may later be removed if $u$ is not among its $k$ top-neighbours. 

We can reduce the memory usage of the algorithm by adding $u$ to $A(v)$ only if $\pi(u)< \pi(v)$ and adding $v$ to $A(u)$ only if $\pi(v)< \pi(u)$. This change does not affect the resulting clustering because $u$ is always added to $A(u)$, and, consequently, we never pick a pivot $v$ with $\pi(v)>\pi(u)$ at the  clustering phase of the algorithm. Nevertheless, to simplify the exposition, we will assume that we always add $u$ to $A(v)$ and $v$ to $A(u)$.

In this paper, we will call clusters singleton clusters only if they were created in the ``else'' clause of the clustering phase. That is, if they contain a single vertex $u$, which is not a pivot. Note that some other clusters may contain a single vertex as well. We will not call them singleton clusters.

Before analyzing the approximation factor of the algorithm, we briefly discuss its memory usage and running time.
For every vertex $u$, we keep its rank $\pi(u)$, a set of its $k$ highest ranked  neighbors  $A(u)$, a pointer to the pivot if $u$ is not a singleton, and two bits indicating whether $u$ is a pivot and $u$ is a singleton cluster. Thus, the total memory usage is $O(kn)$ words; each word is an integer that can store a number between $1$ and $n$.

\enlargethispage*{0.3cm}

The running time of the initialization step is $O(n)$. The running time of the streaming phase is $O(m\log k)$ where $m$ is the number of positive edges (the algorithm needs to perform a constant number of priority queue operations for each edge in the stream; the cost of each operation is $O(\log k)$). The running time of the clustering phase is $O(kn)$, since we need to examine all neighbors $v$ of $u$ in queue $A(u)$. We think of $k=O(\nicefrac{1}{\varepsilon})$ as being a constant.

\subsection{Approximation Factor}
We now show that the approximation factor of the algorithm is $3+O(1/k)$. For the sake of analysis, we consider an equivalent variant of our algorithm that reveals the ordering $\pi$ one step at a time. This algorithm is not a streaming algorithm. However, it produces exactly the same clustering as our single-pass streaming algorithm. Thus, it is sufficient to show  that it has a $3+O(1/k)$ approximation. The algorithm, which we will refer to as \emph{Algorithm 2}, works as follows. Initially, it marks all vertices as unprocessed and not clustered. It also sets a counter $K_1(u)=0$ for each $u$. We denote the set of all unprocessed vertices at step $t$ by $V_t$ and the set of all non-clustered vertices by $U_t$. Then, $V_1=V$ and $U_1 = V$. 
At step $t$, the algorithm picks a random vertex $w_t$ in $V_t$. If $w_t$ is not yet clustered ($w_t\notin U_t$), then $w_t$ is marked as a pivot, a new cluster is created for $w_t$, and all not yet clustered neighbours of $w_t$ including $w_t$ itself are added to the new cluster. If $w_t$ is already clustered, then the counter $K_t(v)$ is incremented by $1$ for all not yet clustered neighbours $v$ of $w_t$ (i.e., $K_{t+1}(v)=K_{t}(v) +1$). After that, all vertices $v$, whose counter $K_{t+1}(v)$ got equal to $k$ at this step, are put into singleton clusters. We mark $w_t$ as processed (by letting $V_{t+1}=V_t\setminus \{w_t\}$) and mark all vertices clustered at this step as clustered (by removing them from $U_t$).

We show that this algorithm outputs the same clustering as Algorithm 1. To this end, define an ordering $\pi$ for Algorithm 2 as $\pi: w_t \mapsto t$, where $w_t$ is the vertex considered at step $t$ of Algorithm 2. Clearly, $\pi$ is a random (uniformly distributed) permutation. We prove the following lemma.

\begin{lemma}
If  Algorithm 1
and Algorithm 2 use the same ordering $\pi$, then they produce exactly the same clustering of $V$.
\end{lemma}
\begin{proof}
We show that every vertex $u$ is assigned to the same pivot $v$ in both algorithms, or $u$ is put in a singleton cluster by both algorithms assuming that Algorithms 1 and 2 use the same ordering  $\pi$. The proof is by induction on the rank $\pi(u)$. Consider a vertex $u$ and assume that all vertices $v$ higher ranked than $u$ are clustered in the same way by Algorithms~1 and~2. Both algorithms examine neighbours of $u$ in exactly the same order, $\pi$. Moreover, they only 
consider the top $k$ neighbours: Algorithm 1 does so because it stores  neighbours of $u$ in a  priority queue of size $k$; Algorithm 2 does so because it keeps track  
of the number, $K_t(u)$, of examined neighbours of $u$ and once that number equals $k$, it puts $u$ in a singleton cluster. Both algorithms assign $u$ to the first found pivot vertex $v$ or to a singleton cluster if no such $v$ exists among the top $k$ neighbours of $u$. Hence, both algorithms cluster $u$ in the same way.
\end{proof}

We now analyze Algorithm 2. We split all steps of Algorithm 2 into \emph{pivot} and \emph{singleton} steps. We say that step $t$ is a pivot step if vertex $w_t$ which was processed at this step was not clustered at the beginning of the step (i.e.,
$w_t\in U_t$). We say that $t$ is a singleton step, otherwise (i.e.,
$w_t\in V_t\setminus U_t$). Observe that all clusters created at pivot steps are non-singleton clusters (even though some of them may contain a single vertex; see above for the definition of singleton clusters). In contrast, all clusters created at singleton steps are singleton clusters. We will now show that the expected cost of clusters created at pivot steps is at most $3OPT$. Then, we will analyze singleton steps and prove that the cost of clusters created at those steps is at most $O(1/k)\,OPT$, in expectation (where $OPT$ is the cost of the optimal solution). 

We say that a positive or negative edge $(u,v)$ is settled at step $t$ if both $u$ and $v$ were not clustered at the beginning of step $t$ but $u$, $v$, or both $u$ and $v$ were clustered at step $t$. Note that once an edge $(u,v)$ is settled, we can determine if $(u,v)$ is in agreement or disagreement with the clustering produced by the algorithm:
For a positive edge $(u,v)\in E^+$ if both $u$ and $v$ belong to the newly formed cluster, then $(u,v)$ is in agreement with the clustering. For a negative edge $(u,v)\in E^-$ if both $u$ and $v$ belong to the newly formed cluster, then $(u,v)$ is in disagreement with the clustering. Similarly, if only one of the endpoints $u$ or $v$ belongs to the new cluster, then $(u,v)$ is in agreement with the clustering for a negative edge $(u,v)\in E^-$ and in disagreement with the clustering for a positive edge $(u,v)\in E^+$. 

We say that the cost of a settled edge is $0$ if it agrees with the clustering and $1$ if it disagrees with the clustering. Let $P$ be the cost of all edges settled during pivot steps. Furthermore, let $P^+$ be the cost of all positive edges and $P^-$ be the cost of all negative edges settled during pivot steps. Then, $P= P^+ + P^-$. If $v$ is a singleton cluster, let $S(v)$ be the number of positive edges incident to $v$ the algorithm cut when it created this cluster. Note, that some edges incident to $v$ might be cut before $v$ was put in a singleton cluster. We do not count these edges in $S(v)$. 

Every edge which is in disagreement with the clustering is (1) a positive edge cut at a pivot step of the algorithm; (2) a negative edge joined at a pivot step of the algorithm; or (3) a positive edge cut at a singleton step of the algorithm. Hence, the cost of the algorithm equals:
$$ALG= P^+ + P^- + \sum_{v \in V}S(v).$$

\begin{lemma}
We have $\E[P^+ + P^-]\leq 3OPT$.
\end{lemma}
\begin{proof}
Observe that pivot steps are identical to the steps of the 
\textsc{Pivot} algorithm. That is, given that $w_t$ belongs to $U_t$, the conditional distribution of $w_t$ is uniform in $U_t$ and the cluster created at this step contains $w_t$ and all its positive neighbours. 
Let us fix an optimal solution and denote its cost by $OPT$. In this proof, we also refer to the cost of a set of edges in the fixed optimal solution as the optimal cost of edges in that set. \cite*{ailon2008aggregating} proved that the expected cost of the solution produced by the \textsc{Pivot} algorithm is upper bounded by $3OPT$. Furthermore, they showed that the expected cost of edges settled at step $t$ is upper bounded by $3\E[\Delta OPT_t]$, where $\Delta OPT_t$ is the optimal cost of edges settled by \textsc{Pivot} at step $t$ (see also Lemma 4 in \cite{chawla2015near}).  Therefore, the expected total cost of all edges settled at pivot steps of our algorithm is upper bounded by $3\E[\sum_{t \text{ is pivot step}}\Delta OPT_t]$. The optimal cost of all settled edges at pivot steps is upper bounded by the optimal cost of all edges, which equals $OPT$. Hence, $\E[P^+ + P^-]\leq 3OPT$.
\end{proof}

We now bound $\E[\sum_v S(v)]$. Fix a vertex $v$. Let $N(v)$ be the set of all vertices connected with $v$ by a positive edge (``positive or similar neighbours''), and $D^+(v)=|N(v)|$ be the positive degree of $v$. We now define a new quantity $X_t(v)$. 
If $v$ is not in a singleton cluster at step $t$, then let  $X_t(v)$ be the number of positive edges incident on $v$ cut in the first $t-1$ steps of the algorithm. Otherwise, 
let $X_t(v)$ be the number of positive edges incident on $v$ cut by the algorithm before $v$ was placed in the singleton cluster.
The number of edges incident to $v$ cut by the algorithm equals 
$X_{n+1}(v) + S(v)$. Hence, the total number of cut positive edges equals 
$\dfrac{1}{2}\sum_{v\in V} \big(X_{n+1}(v) + S(v)\big)$. It also equals
$P^+ + \sum_{v\in V} S(v)$. Thus,
$$
\E\Big[\sum_{v\in V} X_{n+1}(v)\Big]
\leq 
\E\big[2P^+ + \sum_{v\in V} S(v)\big]\leq 6OPT + 
\E\Big[\sum_{v\in V} S(v)\Big].
$$
We now show that $\E[S(v)]\leq \nicefrac{1}{k}\;\E[X_{n+1}(v)]$ for all $v$ and, therefore,
$$\E\Big[\sum_{v\in V} S(v)\Big]
\leq \frac{6 OPT}{k-1},
$$
and 
$$\E[ALG]= \E\Big[P^+ + P^- + \sum_{v \in V}S(v)\Big]\leq \Big(3 + \frac{6}{k-1}\Big)\,OPT.$$

To finish the analysis of Algorithm 2, it remains to prove the following lemma.

\begin{lemma}
We have $\E[S(v)]\leq \nicefrac{1}{k}\;\E[X_{n+1}(v)]$.
\end{lemma}
\begin{proof}  
Define the following potential function:
$$\Phi_t(v) = 
\textbf{1}(v\notin U_t)\cdot X_t(v) - K_t(v)\cdot (D^+(v)-X_t(v)),
$$
where $\textbf{1}(v\notin U_t)$ is the indicator of event $\{v\notin U_t\}$. 
Let $\calF_t$ be the state of the algorithm at the beginning of step $t$. 
We claim that $\Phi_t(u)$ is a submartingale i.e., 
$\E[\Phi_{t+1}(u)\mid \calF_t]
\geq \Phi_t(u)$. 
\begin{claim}
$\Phi_t(v)$ is a submartingale.
\end{claim}
\begin{proof}
Consider step $t$ of the algorithm. If $v$ is already clustered (i.e., $v\notin U_t$), then all edges incident on $v$ are settled and terms $\textbf{1}(v\notin U_t)$, $K_t(v)$, and $X_t(v)$ no longer change. Hence, the potential function $\Phi_t(v)$ does not change as well.

Suppose that $v\in U_t$. If at step $t$, the algorithm picks $w_{t}$ not from the neighbourhood $N(v)$, then $v$ does not get clustered at this step and also $K_t(v)$ does not change. The value of $X_t(v)$ may increase as some neighbours of $v$ may get clustered. However, $X_t(v)$ cannot decrease. Therefore, 
$\Phi_{t+1}(v)\geq \Phi_{t}(v)$, and $\E\big[\Phi_{t+1}(v) - \Phi_{t}(v)\mid 
w_{t}\notin N(v)
;\;
\calF_t\big]
\geq 0$.

Let us now assume that at step $t$, the algorithm chooses $w_{t}$ from the neighbourhood of $v$. Then, $v$ is clustered with conditional probability at least
$$
\Pr\Big(w_{t}\in N(v) \cap U_t\mid w_{t}\in N(v)
; \calF_t\Big) = 
\frac{|N(v)\cap U_t|}{|N(v)\cap V_t|},
$$
because if the pivot is chosen in $N(v) \cap U_t$, then $v$ is put in the cluster of $w_t$. We have $|N(v) \cap U_t| = D^+(v) - X_t(v)$, since every clustered neighbour of $v$ contributes $1$ to the number of cut edges $X_t(v)$, and the total number of $v$'s positive neighbours equals $D^+(v)$. Also, note that $|N(v)\cap V_t|\leq |N(v)|= D^+(v)$. Hence,
$$
\Pr\Big(w_t\in N(v)\cap U_t\mid w_{t}\in N(v)
; \calF_t\Big) \geq 
\frac{D^+(v) - X_t(v)}{D^+(v)}.
$$
With this conditional probability, $\mathbf{1}(v\notin U_t)$ gets equal to $1$ at step $t$, and, as a result, $\Phi_{t}(v)$ increases by $X_{t+1}(v)\geq X_{t}(v)$. 
With the remaining probability, $K_t(v)$ increases by $1$, and, as a result, 
$\Phi_{t}(v)$ decreases by at most
$D^+(v) - X_{t+1}(v)\leq D^+(v)-X_{t}(v)$.
Thus, 
\begin{multline*}
\E\big[\Phi_{t+1}(v) - \Phi_{t}(v)\mid 
w_{t}\in N(v)
;\;
\calF_t\big]
\geq 
\frac{D^+(v) - X_t(v)}{D^+(v)} \cdot X_t(v)
-
\frac{X_t(v)}{D^+(v)} \cdot (D^+(v) - X_t(v)) = 0
.
\end{multline*}
This concludes the proof.
\end{proof}

Since $\Phi_1(v) = 0$ and $\Phi_t(v)$ is a submartingale, we have $\E[\Phi_{n+1}(v)] \geq 0$. 
Therefore,
$$
\E[X_{n+1}(v)]
=
\E[\textbf{1}(v\notin U_{n+1})\cdot X_{n+1}(v)]
\geq  
\E[K_{n+1}(v)\cdot (D^+(v)-X_{n+1}(v))].$$
If $v$ is in a singleton cluster, then $K_{n+1}(v)= k$ and $X_{n+1}(v) + S(v) = D^+(v)$ (because all positive edges incident to $v$ are cut). If $v$ is not in a singleton, then $S(v)=0$. Hence,
$k S(v) \leq K_{n+1}(v)\cdot (D^+(v)-X_{n+1}(v))$. Consequently, 
$\E[X_{n+1}(v)]\geq k \E[S(v)]$.
\end{proof}

\bibliographystyle{abbrvnat}
\bibliography{references}
\end{document}